\documentclass[11pt,epsf]{amsart}
\usepackage{enumerate}
\usepackage{latexsym}
\usepackage{amsmath}
\usepackage{amssymb}
\usepackage{amsthm}
\usepackage{epsfig}
\usepackage[latin1]{inputenc}
\usepackage{color}

\newtheorem{theorem}{Theorem}[section]
\newtheorem{lemma}[theorem]{Lemma}
\newtheorem{corollary}[theorem]{Corollary}
\newtheorem{claim}[theorem]{Claim}
\newtheorem{observation}[theorem]{Observation}

    \def\independenT#1#2{\mathrel{\setbox0\hbox{$#1#2$}%
    \copy0\kern-\wd0\mkern4mu\box0}}

\def\opt{\text{\sc opt}}

\def\G{\mathcal{G}}

\def\cal{\mathcal}

\def\G{\cal G}
\def\P{\cal P}
\def\S{\cal S}

\def\opt{\text{\sc opt}}

\begin{document}
\addtocounter{page}{-1}

\bibliographystyle{plain}
\renewcommand{\thefootnote}{\fnsymbol{footnote}}

\begin{center}
{\Large \bf On the Implications of Lookahead Search in Game Playing}\\ [1cm]

{\sc Vahab Mirrokni}\footnote{Google Research, Google. Email: {\tt mirrokni@google.com} },
{\sc Nithum Thain}\footnote{Oxford-Man Institute, University of Oxford.
Email: {\tt nithum@gmail.com} }
and
{\sc Adrian Vetta}\footnote{Department of Mathematics and Statistics, 
and School of Computer Science, McGill University. 
Email: {\tt vetta@math.mcgill.ca} }\\[.7cm]

\today\\[1cm]

\end{center}

\renewcommand{\thefootnote}{\arabic{footnote}}
\begin{center}
\begin{quote} 
{ \small
{\bf Abstract.}
Lookahead search is perhaps the most natural and widely used game playing strategy.
Given the practical importance of the method, the aim of this paper is to provide a
theoretical performance examination of lookahead search in a wide variety of applications.

To determine a strategy play using {\em lookahead search}, each agent predicts multiple levels of possible re-actions to 
her move (via the use of a search tree), 
and then chooses the play that optimizes her future payoff accounting for these re-actions.
There are several choices of optimization function the agents can choose, where the most appropriate choice of function 
will depend on the specifics of the actual game - we illustrate this in our examples.
Furthermore, the type of search tree chosen by computationally-constrained agent can vary.  
We focus on the case where agents can evaluate only a bounded number, $k$, of moves into the future.
That is, we use depth $k$ search trees and call this approach {\em k-lookahead search}.

We apply our method in five well-known settings: AdWord auctions; industrial organization (Cournot's model); 
congestion games; valid-utility games and basic-utility games; cost-sharing network design games.
We consider two questions. First, what is the expected social quality of outcome when agents apply lookahead search?
Second, what interactive behaviours can be exhibited when players use lookahead search? 
 
Myopic game playing (whose corresponding equilibria are Nash equilibria), where each player can only foresee the 
immediate effect of her own actions, is the special case of $1$-lookahead search. Thus, for the first question, it is natural
to ask whether social outcomes improve when players use more foresight than in myopic behaviour. 
The answer depends on the game played: \\
\noindent(i) In Adword auctions (or generalized second-price auctions), we show that $2$-lookahead game playing 
results in outcomes that are always optimal to within a constant factor; in contrast, myopic game play can produce 
arbitrarily poor equilibrium outcomes. \\
\noindent(ii) For the Cournot game, applying $2$-lookahead leads to a $12.5\%$ increase in output and a $5.5\%$ 
increase in social surplus compared with myopic competition. Similar bounds arise as the length $k$ of foresight
increases.\\
\noindent(iii) For congestion games, as with myopic game playing, 
lookahead search leads to constant factor qualitative guarantees.\\
\noindent(iv) For basic-utility games, on the other hand, whilst myopic game playing always leads to 
constant factor approximations, additional foresight can lead to arbitrarily bad solutions!\\
\noindent(v) In a simple Shapley network design game, qualitative guarantees improve with the
length of foresight.\\
Regarding the second question, a variety of interesting game playing characteristics also arise with lookahead search.
Stackelberg leader-follower behaviours can be induced when the players have asymmetric computational power.
For example, Stackelberg equilibria can be produced in the Cournot game.
Lookahead search can also generate ``uncoordinated" cooperative behaviour! 
An example of this is shown for the Shapely network design game.

}
\end{quote}
\end{center}

\renewcommand{\theenumi}{(\alph{enumi})}
\renewcommand{\labelenumi}{(\alph{enumi})}

\addtocounter{footnote}{-3}

\renewcommand{\theenumi}{(\alph{enumi})}
\renewcommand{\labelenumi}{(\alph{enumi})}

\newpage

\section{Introduction}\label{sec:intro}

Our goal here is not to prescribe how games should be played. Rather, we wish to analyse
how games actually are played. To wit we consider the strategy of lookahead search,
described by Pearl \cite{Pea84}  in in his classical book on heuristic search as being used by
 ``almost all game-playing programs". 
To understand the lookahead method and the reasons for its ubiquity in practice, consider an agent trying to decide
upon a move in a game. Essentially, her task is to evaluate each of her possible moves
(and then select the best one). Equivalently, if she know the values of each child node 
in the game tree then she can calculate the value of the current node.
However, the values of the child nodes may also be unknown! 
Recall two prominent ways to deal with this. Firstly,  crude estimates based upon local information could be used 
to assign values to the children; this is approach taken by {\em best response dynamics}. 
Secondly, the values of the children can be determined recursively by finding the values of the 
grandchildren. At its computation extreme, this latter approach in a finite game is Zermelo's algorithm
- assign values to the leaf nodes\footnote{Often the values of the leaf nodes will be true values rather than 
estimates, for example when they correspond to end positions in a game.} 
of the game tree and apply backwards induction to
find the value of the current node. 

Both these approaches are special cases of {\em lookahead search}: choose a 
local search tree $T$ rooted at the current node in the game tree; 
valuations (or estimates thereof) are given to leaf nodes of $T$; valuations for internal
tree nodes are then derived using the values of a node's immediate descendants 
via backwards induction; a move is then selected corresponding to the value assigned the root. 
For best response dynamics the search tree is simply the star
graph consisting of the root node and its children. With unbounded computational power, 
the search tree becomes the complete (remaining) game tree used by Zermelo's algorithm.

We remark that the actual shape of the search tree $T$ is chosen {\em dynamically}. 
For example, if local information is sufficient to provide a reliable estimate
for a current leaf node $w$ then there is no need to grow $T$ beyond $w$.
If not, longer branches rooted at $w$ need to be added to $T$.
Thus, despite our description in terms of ``backwards induction", lookahead search is a very forward 
looking procedure. Subject to our computational abilities, we search further forward only if we think it will help 
evaluate a game node. Indeed, in our opinion, it is this forward looking aspect that makes lookahead search such a 
natural method, especially for humans and for dynamic (or repeated) games.\footnote{
In contrast, strategies that are prescribed by axiomatic principles, equilibrium constraints, or notions of regret are 
much less natural for dynamic game players.}

Interestingly, the lookahead method was formally proposed as long ago as 1950 by Shannon \cite{Sha50},
who considered it a practical way for machines to tackle complex problems that
require ``general principles, something of the nature of judgement, and considerable trial and error, rather than 
a strict, unalterable computing process". To illustrate the method, Shannon described in detail
how it could be applied by a computer to play chess. 
The choice of chess as an example is not a surprise: as described the lookahead approach is particularly suited to game-playing.
It should be emphasised again, however, that this approach is natural for all computationally constrained agents,
not just for computers. Lookahead search is an instinctive strategic method utilised by human beings as well.
For example, Shannon's work was in part inspired by De Groot's influential psychology thesis  \cite{Gro46} on human chess
players. De Groot found that all players (of whatever standard) used essentially the same thought process 
- one based upon a lookahead heuristic. Stronger players were better at evaluating positions
and at deciding how to grow (prune or extend) the search tree but the underlying approach was always the same. 

Despite its widespread application, there has been little theoretical examination
of the consequences of decision making determined by the use of local search trees.
The goal of this paper is to begin such a theoretical analysis.
Specifically, what are the quantitative outcomes and dynamics in various games when players use 
lookahead search?

\subsection{{\sc Lookahead Search: The Model.}}\label{sec:conc}\ \\
Having given an informal presentation, let's now formally describe the lookahead method. 
Here we consider games with sequential moves that have complete information.
These assumptions will help simplify some of the underlying issues, but the lookahead approach
can easily be applied to games without these properties.

We have a strategic game $G(\P, \S, \{\alpha_i: i\in \P\})$. Here $\P$ is the set of $n$ players, 
$S_i$ is the set of possible strategies for $i\in \P$, $\S= (S_1\times S_2 \ldots \times S_n)$ is the strategy 
space, and $\alpha_i: \S \rightarrow  R$ is the payoff function for player $i\in \P$. A {\em state} $\bar{s}= (s_1, s_2, \ldots, s_n)$
is a vector of strategies $s_i\in S_i$ for each player $i\in \P$. 

Suppose player $i\in \P$ is about to decide upon a move.
Recall, with lookahead search, she wishes to assign a value to her current state node $\bar{s} \in \S$ that corresponds 
to the highest value of a child node. To do this she selects a search tree $T_i$ over the set of states of the game
rooted at $\bar{s}$.
For each leaf node $\bar{l}$ in $T_i$, player $i$ then assigns a valuation $\Pi_{j,\bar{l}}=\alpha_j(\bar{l})$ 
for each player $j$. 
Valuations for internal nodes in $T_i$ are then calculated by induction as follows: if player $p$ is destined to move 
at game node $\bar{v}$ then his valuation of the node is given by
$$\Pi_{p,\bar{v}} = \max_{u\in \mathcal{C}(\bar{v})} [ r_{p,\bar{v}} + \Pi_{p,\bar{u}}].$$
Here, $\mathcal{C}(\bar{v})$ denotes the set of children of $\bar{v}$ in $T_i$, and $r_{p,\bar{v}}$ 
is some additional payoff received by player $p$  at node $\bar{v}$.
Should $p$ choose the child $\bar{u}^*\in \mathcal{C}(\bar{v})$ then assume any non-moving player $j\neq p$
places a value of $\Pi_{j,\bar{v}} = r_{j,\bar{v}} + \Pi_{j,{\bar{u}^*}}$ on node $\bar{v}$. 
Then given values for children of the root node $\bar{s}$ of $T_i$, player $i$ is thus able to compute the lookahead payoff 
 $\Pi_{i,\bar{s}}$ which she uses to select a move to play at $\bar{s}$.
[The method is defined in an analogous manner if players seek to minimise rather than maximise
their "payoffs".]

After $i$ has moved, suppose player $j$ is then called upon to move.
He applies the same procedure but on a local search tree $T_j$ rooted at the new game node.
Note that $j$'s move may {\bf not} be the move anticipated by $i$ in her analysis. 
For example, suppose all the players use $2$-lookahead search. Then player $i$ calculates on the basis 
that player $j$ will use a $1$-lookahead search tree $T'_j$ when he moves -- because for computational 
purposes it is necessary that $T'_j\subseteq T_i$. But when he moves player $j$ actually uses 
the $2$-lookahead search tree $T_j$ and this tree goes beyond the limits of $T_i$.

\subsection{{\sc Lookahead Search: The Practicalities.}}\label{sec:practical}\ \\
Observe that there is still a great deal of flexibility in how the players implement the model. 
This versatility, we would argue, is a major strength (and another reason
underlying its ubiquity) and not a weakness of the method. For example, 
it accords well with Simon's belief, discussed in Section \ref{sec:related}, that behaviours should be adaptable.
We now give some examples of this adaptability and highlight those aspects that we analyse in this paper.

\noindent$\bullet$ {\bf \noindent Dynamic Search Trees.} 
Recall that search trees may be constructed dynamically. Thus, the exact shape of the search tree utilized will be heavily 
influenced by the current game node, and the experience and learning abilities of the players. 
Whilst clearly important in determining gameplay and outcomes, these influences are a distraction from our focal point, namely,
computation and dynamics in games in which players use lookahead search strategies.
Therefore, we will simply assume here that each $T_i$ is a breadth first search tree of depth $k_i$. Implicitly, $k_i$ is dependent 
on the computational facilities of player $i$.

\noindent$\bullet$ {\bf Evaluation Functions.}
Different players may evaluate leaf nodes in different ways. To evaluate internal nodes, as described above, we make the 
standard assumption that they use a $\max$ (or $\min$) function.
This need not be the case. For example, a risk-averse player may give a higher value to a node (that it does not own)
with many high value children than to a node with few high value children -- we do not consider such players here. 

\noindent$\bullet$ {\noindent \bf Internal Rewards or Not: Path Model vs Leaf Model.}
We distinguish between two broad classes of game that fit in this framework but are conceptually quite different. 
In the first category, payoffs are determined only by outcomes
at the end of game. Valuations at leaf nodes in the local search trees are then just estimates of the what the final outcome 
will be if the game reaches that point. Clearly chess falls into this category.
In the second category, payoffs can be accumulated over time - thus different paths with the same endpoints 
may give different payoffs to each player. Repeated games, such as industrial games over multiple time periods, 
can be modelled as a single game in this category. The first category is modelled by setting all internal 
rewards $r_{p,\bar{v}}=0$. Thus what matters in decision making
is simply the initial (estimated) valuations a player puts on the leaf nodes.
We call this the {\em leaf (payoff) model} as an agent then strives to reach a leaf of $T_i$ with as high a value as possible.
The second category arises when the internal rewards, $r_{p,\bar{v}}$, can be non-zero. Each agent then wishes to traverse
paths that allow for high rewards along the way.
More specifically, in this model, called the {\em path (payoff) model}, the internal reward is $r_{p,\bar{v}} = \alpha_p(\bar{v})$.

\noindent$\bullet$ {\noindent \bf Order of Moves: Worst-Case vs Average-Case.}
In multiplayer games, the order in which the players move may not be fixed.
This adds additional complexity to the decision making process, as the local search tree
will change depending upon the order in which players move. Here, we will
examine two natural approaches a player may use in this situation:  
{\em worst case lookahead} and {\em average case lookahead}.
In the former situation, when making a move, a risk-averse player will assume
that the subsequent moves are made by different players chosen by an adversary to minimize
that player's payoff.
In the latter case, the player will assume that each subsequent move is made
by a player chosen uniformly at random; we allow players to make consecutive moves. 
In both cases, to implement the method the player must perform calculations for multiple search trees.
This is necessary to either find the worst-case or perform expectation calculations.

\subsection{{\sc Techniques and Results.}}\label{sec:results}\ \\
We want to understand the social quality of outcomes that arise when computationally-bounded agents 
use $k$-lookahead search to optimise their {\em expected} or {\em worst-case} payoff 
over the next $k$ moves. Two natural ways we do this are via {\bf equilibria} and via the study of {\bf game dynamics}.
To explain these approaches, consider the following definition.
Given a lookahead payoff function, $\Pi_{i,\bar{s}}$, a
{\em lookahead best-response} move for player $i$, at a state $\bar{s}\in \mathcal{S}$,
is a strategy $s_i$ maximising her lookahead payoff, that is,  $\forall s'_i\in S_i$: 
$\Pi_{i,\bar{s}} \ge \Pi_{i,(\bar{s}_{-i},s'_i)}$.
[A move $s'_i$ for player $i$, at a state $\bar{s} \in \mathcal{S}$,
is {\em lookahead improving} if $\Pi_{i,\bar{s}} \le \Pi_{i,(\bar{s}_{-i},s'_i)}$.]
A {\em  lookahead equilibrium} is then a collection
of strategies such that each player is playing her lookahead best-response move
for that collection of strategies. Our focus here is on pure strategies.
Then, given a social value for each state, the {\em coordination ratio} (or price of anarchy) {\em of lookahead equilibria} 
is the worst possible ratio between the social value of a lookahead equilibrium and the optimal global social value. 

To analyse the dynamics of lookahead best-response moves, we examine the expected social value of 
states on polynomial length random walks on the
{\em lookahead state graph}, $\G$. This graph has a node for each state $s\in \S$
and an edge from $\bar{s}$ to a state $\bar{t}$ with a label $i\in \P$ if the only difference between $\bar{s}$ and $\bar{t}$ 
is that player $i$ changes strategy from $s_i$ to $t_i$, where $t_i$ is the lookahead best response move at $\bar{s}$. 
The {\em coordination ratio of lookahead dynamics}
is the worst possible ratio between the expected social value of states on a polynomially long random walk on $\G$
and the optimal global social value. 

For practical reasons, we are usually more interested in the dynamics of lookahead best-response moves than
in equilibria. For example, as with other equilibrium concepts, lookahead best-response moves may not lead 
to lookahead equilibria. Indeed, such equilibria may not even exist.
Typically, though, the methods used to bound the coordination ratio for $k$-lookahead equilibria can be combined with 
other techniques to bound the coordination ratio for $k$-lookahead dynamics. We show how to do this
for congestions games in Section \ref{sec:SR}; see also Goemans et al. \cite{GMV05} for several examples
with respect to $1$-lookahead dynamics. Consequently, for both simplicity and brevity, most of the results we
give here concern the coordination ratio for lookahead equilibria. We are particularly interested in discovering when 
lookahead equilibria guarantee good social solutions, and how outcomes vary with different levels of foresight ($k$).
We perform our analyses for an assortment of games including an AdWord auction game, the Cournot game,
congestion games, valid-utility games, and a cost-sharing network design game. 

We begin, in Section \ref{sec:ad}, by considering strategic bidding in an AdWord generalised second-price auction, 
and studying the social values of the allocations in the resulting equilibria. 
In particular, we show that $2$-lookahead game playing results in the optimal outcome or a constant-factor 
approximate outcome under the leaf and path models, respectively. This is in contrast to $1$-lookahead (myopic) game playing which
can result in arbitrarily poor equilibrium outcomes, and shows that more forward-thinking bidders
would produce efficient outcomes.

Second, in Section \ref{sec:Cournot}, we examine the Cournot duopoly game. 
Here two firms compete in producing a good consumed by a set of buyers via the choice of production quantities.
We study equilibria of these simple games resulting from
$k$-lookahead search. The equilibria of these simple games for myopic game playing, $k=1$, is well-understood.
For $k>1$, however, firms produce over $10\%$ more than if they were competing myopically; this is better for society 
as it leads to
around a $5\%$ increase in social surplus. Surprisingly, the optimal level of foresight for society is $k=2$.
Furthermore, we show that
Stackelberg behaviours arise as a special case of lookahead search where the firms have asymmetric computational abilities.

Third, in Section \ref{sec:SR}, we examine congestion games with linear latency functions, and study 
the average of delay of players in those games. We show that $2$-lookahead game playing
results in constant-factor approximate solutions. In particular, the coordination ratio of lookahead 
dynamics is a constant. These guarantees are similar to those obtained via $1$-lookahead.

Fourth, in Section \ref{sec:UG}, we consider two classes of resource sharing games, known as 
valid-utility and basic-utility games. For both of these games,
we show that lookahead game playing may result in very poor solutions.
For valid-utility games, we show $k$-lookahead can give a coordination ratio for lookahead dynamics
of $\Theta(\sqrt{n})$. Myopic game play can also give very poor solutions \cite{GMV05}, but 
additional foresight does not significantly improve outcomes in the worst case. For basic-utility games,
however, myopic game dynamics give a constant coordination ratio \cite{GMV05} whereas 
we show that $2$-lookahead game playing may result in $o(1)$-approximate social welfare with the leaf model. 
Thus, additional foresight in games need not lead to better outcomes, as is traditionally assumed in decision theory.

Finally, in Section \ref{sec:Shapley}, we present a simple example of a cost-sharing network design game that illustrates
how the use of lookahead search can encourage cooperative behaviour (and better outcomes) 
{\em without} a coordination mechanism.

Observe that our results show that lookahead search has different effects depending upon the game. 
It would be interested to study further which game structures lead to more beneficial outcomes when 
longer foresight is used, and which game structures lead to more detrimental outcomes.

\subsection{{\sc Background and Related Work.}}\label{sec:related}\ \\
This work is best viewed within the setting of {\em bounded rationality} pioneered by
Herb Simon. In Rational Choice Theory a {\em rational} agent (or economic man) makes decisions via utility maximisation.
Whilst the non-existence of economic man is not in doubt, rationality remains a central assumption in economic thought.
This is typically justified using an {\em as if} as expounded by Friedman \cite{Fri53}:
whether people are actually rationality or not is unimportant provided their actions can be viewed in a way that is consistent 
with rational decision making - that is, provided agents act as if they are rational.\footnote{{\small For example, a 
consumer whose purchasing strategy allocates fixed proportions of her budget to specific goods (regardless of price levels) 
can be viewed as rational consumer with a Cobb-Douglas utility function!}}
Friedman concluded that a model should be judged by it predictive value rather than by the
realism of its assumptions. On this scale rationality often (but not always) does very well.

However, motivated by considerations of computational power and predictive ability,
Simon \cite{Sim55} argued that ``the task is to replace the global rationality
of economic man with a kind of rational behaviour
that is compatible with the access to information
and the computational capacities that are actually
possessed by organisms, including man, in the kinds
of environments in which such organisms exist".
He argued that, instead of optimising, agents apply heuristics in 
decision making. An example of this being the {\em satisficing} heuristic:
agents search for feasible solutions, stopping when then discover
an outcome that achieves an aspired level of satisfaction\footnote{{\small Over time,
and depending upon what is found in the search,
this aspiration level may be changed.}}. 
We remark that the use of a search phase provides a fundamental distinction between rational
and boundedly rational agents. For rational agents the search is irrelevant as they will anyway make an 
optimal choice given the constraints of the problem. For agents of bounded rationality the form of the search can 
heavily influence decision making.

Interestingly, De Groot's work on chess players also heavily influenced Simon's general thinking on 
cognitive science.\footnote{{\small In fact, Simon sent his student George Baylor to help translate De Groot's work into English.}}
This is exemplified in his famous book with Newell on human problem solving \cite{NS72}, where 
humans are viewed as information processing systems.

The label bounded rationality is currently used in a number of disparate areas some of which actually go against 
the main thrust of Simon's original ideas; see Selten \cite{Sel01} and Rubenstein \cite{Rub98} for some discussion on this point.
Two schools of thought developed by psychologists, experimental economists, and behavioural economists 
are, however, well worth mentioning here. First, the {\em Heuristics and Biases} program espoused by Kahneman and Tversky
and, second, the {\em Fast and Frugal Heuristics} program espoused by Gigerenzer.
Whilst both programs agree that humans routinely use simple heuristics in decision making,
their philosophical outlooks are very different. The former program primarily looks for outcomes (caused by the use of heuristics) in violation of 
subjective excepted utility theory, and views such biases as a sign of irrationality
likely to lead to poor decision making. In contrast, the latter program views the use of heuristics
as natural and, in principle, entirely compatible with good decision making. 
For example, simple heuristics may 
be more robust to environmental changes and actually outperform methods based upon subjective excepted utility maximisation.
As with the work of Simon, for the fast and frugal heuristics school, the actual quality of an heuristic is assumed to be dependent upon the 
search - how to search and when to stop searching - and the choice of decision rule after the search is terminated. 
Clearly, the lookahead heuristic can be viewed in this light: there is a search (via a local search tree), there is a ``stopping rule"
(determined, for example, by computational constraints and by the expertise of the player), and there is a decision rule (backwards induction).  

The value of lookahead search in decision-making has been examined by the artificial intelligence community \cite{Nau83a}; for examples in effective diagnostics and real-time planning see  \cite{KRS92} and \cite{SKN09}. Lookahead search
is also related to the sequential thinking framework in game theory \cite{Nag95, SW94}.
However, compared to these works and the research carried out by the two schools above,
our focus is more theoretical and less experimental and psychological.
Specifically, we desire quantitative performance guarantees for our heuristics.

Our research is also related to works on the price
of anarchy in a game, and convergence of game dynamics to approximately optimal solutions \cite{MV04,GMV05} 
and to sink equilibria \cite{GMV05,FP08}. Numerous articles study the convergence rate of best-response dynamics 
to approximately optimal solutions \cite{CMS06,FFM08,AAEMS08,BHLR08}.  For example, polynomial-time
bounds has been proven for the speed of convergence to approximately optimal solutions
for approximate Nash dynamics in a large class of potential games~\cite{AAEMS08},
and for learning-based regret-minimisation dynamics for valid-utility games~\cite{BHLR08}.
Our work differs from all the above as none of them capture lookahead dynamics.
In another line of work, convergence of best-response dynamics to (approximate) equilibria and the complexity
of game dynamics and sink equilibria have been
studied~\cite{FPT04,ARV06a,CS07,SkopalikV08,FP08,MS09}, but our paper does not
focus on these types of dynamics or convergence to equilibria.

Motivated by concerns of stability, convergence, and predictability of equilibria
and game dynamics, various equilibrium concepts other than Nash equilibria have been studied
in the economics literature. Among them are correlated equilibria~\cite{A74},
stable equilibria~\cite{KM86}, 
stochastic adjustment models~\cite{KMR93}, 
strategy subsets closed under rational behaviour (CURB set)~\cite{BW90},
iterative elimination of dominated 
strategies, the set of undominated strategies, etc. 
Convergence and strategic stability of equilibria in evolutionary game theory
is also an important subject of study. 
Many other game-theoretic models have been proposed to capture the 
self-interested behaviour of agents.
As well as best-response dynamics, noisy best-response dynamics~\cite{E93,Y93,MontanariS09}, where players occasionally make
mistakes, and simultaneous Nash dynamics~\cite{BerenbrinkFGGHM06},
where all players change their strategies simultaneously, are both well-studied.
In many other models the effect of learning algorithms~\cite{Y04} is examined, 
for example, regret minimisation dynamics~\cite{FV97,HM00,H05,BlumM07,BEL06,BHLR08,EMN09} 
and fictitious play~\cite{Brown51}. 
In most of these studies the most important factor is the 
stability of equilibria, and not measurements of the social value of equilibria.
Furthermore, most of them are motivated by theoretical game theoretic concepts
rather than practical game-playing, and none of the above works consider lookahead search.

\section{Generalised Second-Price Auctions}\label{sec:ad}

For our first example, we apply the lookahead model to generalised second-price (GSP) auctions.
Our main results are that outcomes are provably good when agents use additional foresight; in contrast, myopic
behaviour can produce very poor outcomes. 

The auction set-up is as follows. There are $T$ slots with click-through rates $c_1 > c_2 > ... > c_T > 0$, that is, higher indexed slots have lower 
click-through rates. There are $n$ players bidding for these slots, each with a private valuation $v^i$. Each player $i$ makes a bid $b^i$.
Slots are then allocated via a {\em generalised second price auction}. 
Denote the $j$th highest bid in the descending bid sequence by $b_j$, with corresponding valuation $v_j$.
The $j$th best slot, for $j\leq T$, is assigned to the $j$th highest bidder who is charged a price equal to $b_{j+1}$.
The $T$ highest bidders are called the ``winners''. 
According to the pricing mechanism, if bidder $i$ were to get slot $t$ in the final 
assignment, then he would get utility $u_t^i = (v^i - b_{t+1})c_t$. We denote a player $i$'s utility if he bids $b^i$ by $u^i(b^i)$ 
(the other players bids are implicit inputs for $u^i$).

This auction is used in the context of keyword ad auctions (e.g, Google AdWords) for sponsored search. 
Given the continuous nature of bids in the GSP auction, the best response of each bidder $i$ for any vector of bids by other bidders
corresponds to a range of bid values that will result in the same outcome from $i$'s perspective.
Among these set of bid values, we focus on a specific bid value $b^i$, called the
{\em balanced} bid~\cite{CDE07}. 
The balanced bid $b^i$ is a best-response bid that is as high as possible such that player $i$ cannot be harmed 
by a player with a better slot undercutting him, i.e. bidding just below him. It is easy to calculate that for player $i$ in 
slot $t$, $1 \leq t < T$, the only balanced bid is 
$$b^i = (1-\frac{c_t}{c_{t-1}}) v^i + \frac{c_t}{c_{t-1}} b_{t+1}. $$ 

An important property of balanced bidding is that each ``losing" player $i$ (one not assigned a slot) 
should bid truthfully, that is $b^i = v^i$. To see this add dummy slots with $c_t=0$ if $t>T$.
The player who wins the top slot should also bid truthfully under balanced bidding.
Balanced bidding is the most commonly used bidding strategy~\cite{CDE07, MT10}. 
For some intuition behind this, note that balanced bidding has several 
desirable properties. For a competitive firm, bidding high obviously increases the chance of 
obtaining a good slot. Within a slot this also has the benefit of pushing up the price
a competitor pays without affecting the price paid by the firm.
On the other hand, bidding high increases the upper bound on the price the firm 
may pay, leading to the possibility that the firm may end up paying a high price for one of 
the less desirable slots. Balanced bidding eliminates the possibility 
that a change in bid from a higher bidder can hurt the firm. (Clearly, it is impossible to obtain such a 
guarantee with respect to a lower bidder.) Thus, balanced bidding provides some of the benefits of
high bidding at less risk.
Balanced bidding naturally converges to Nash equilibria unlike
other bidding strategies such as altruistic bidding or competitor busting \cite{CDE07}. 
Moreover, the other bidding strategies would require 
some discretization of players' strategy space in order to analyse the best response dynamics~\cite{CDE07,MT10}.
Consequently, balanced bidding is the most natural strategy choice for our analysis.

For this auction problem, we consider only the leaf model. The leaf model seems more natural than the path model for a single auction
as players are interested in the final allocation output by the auction (there are no intermediary payoffs).
We analyse both worst-case and average-case lookahead; 
depending upon the level of risk-aversion of the agents both cases seem natural in auction settings.

Let player $i$'s lookahead payoff (or utility) at bid $b^i$ with respect to player $j$, denoted by $u^{ij}(b^i)$, be player $i$'s payoff (or utility) 
after player $j$ makes a best-response move.  
In the worst-case lookahead model, we define player $i$'s lookahead payoff for 
a vector $\bar{b}$ of bids as $\Pi_{i,\bar{b}} = \tilde{u}^i(b^i) = \min_{j} u^{ij}(b^i)$. 
In the average-case lookahead model, player $i$'s lookahead payoff $\Pi_{i,\bar{b}}$ for
a bid vector $\bar{b}$ is $\Pi_{i,\bar{b}}=\bar{u}^i(b^i) = \frac{1}{n} \sum_{j} u^{ij}(b^i)$. 
Changing strategy from bid $b^i$ to bid $\bar{b}^i$ is a \textit{lookahead improving} 
move if lookahead utility increases, i.e., $\bar{u}^i (\bar{b}^i) > \tilde{u}^i (b^i)$.
We are at a {\em lookahead equilibrium} if no player has a lookahead improving move.

It is known that the social welfare of Nash equilibria for myopic game playing can be arbitrarily bad~\cite{CDE07} unless
we disallow over-bidding~\cite{LT10}. Here, we  prove the advantage of additional foresight by showing that $2$-lookahead 
equilibria have much better social welfare. In particular, we show that all such equilibria are optimal in the
worst-case lookahead model, and all such equilibria are constant-factor approximate
solutions in the average-case lookahead model.

\subsection{{\sc Worst-Case Lookahead.}}\ \\
Our proof for the worst-case lookahead model can be seen as a generalisation of the proof of ~\cite{BDQ08} for a slightly different model.
We start by proving a useful lemma in this context. 

\begin{lemma} \label{sLem1}
Consider the worst-case lookahead model with the leaf model. 
Label the players so that player $i$ is in slot $i$, and suppose there is a player $t$ such that $v^t < v^{t+1}$. 
Then player $t$ myopically prefers slot $t+1$ to slot $t$.
\end{lemma}

\begin{proof}
Suppose not. Then, as player $t$ does not myopically prefer slot $t+1$ we have
$$(v_t - b_{t+1})c_t \geq (v_t - b_{t+2})c_{t+1}$$
By definition, $b_{t+1} = v_{t+1} - \frac{c_{t+1}}{c_t}(v_{t+1} - b_{t+2})$. Plugging this in gives
\begin{eqnarray*}
(v_t - b_{t+2})c_{t+1} \le 
\left(v_t - \frac{c_t - c_{t+1}}{c_t}v_{t+1} - \frac{c_{t+1}}{c_t}b_{t+2}\right)c_t 
<  \left(\frac{c_{t+1}}{c_t}v_{t} - \frac{c_{t+1}}{c_t}b_{t+2}\right)c_t
= (v_t-b_{t+2})c_{t+1}
\end{eqnarray*}
Thus we obtain our desired contradiction. Note that the strict inequality above follows directly from the fact that
$v^t < v^{t+1}$.
\end{proof}

An 
equilibrium is \textit{output truthful} if the slots are assigned to the same bidders as they would be if bidders 
were to bid truthfully. It is easy to verify that an an allocation optimizes solcial welfare if and only if it is output truthful.
Thus to prove $2$-lookahead equilibria are socially optimal it suffices to show they are output truthful.

\begin{theorem}
For GSP auctions, any $2$-lookahead equilibrium gives optimal social welfare
in the worst-case, leaf model.
\end{theorem}

\begin{proof}
We proceed by contradiction. Consider a non-output-truthful $2$-lookahead equilibrium. 
Again, label the players so that the player $i$ is in slot $i$.
Amongst all the winning players, take the one with the lowest valuation, $v_i$. First suppose that $v_i$ 
is not amongst the $T$ highest valuations. Then, there is a losing player with a higher value than $v_i$. 
But this player is bidding his value, as a result of balanced bidding. Consequently, player $i$'s utility must be negative,
a contradiction. 

Thus, we may assume that $v_i$ is amongst the $T$ highest valuations; specifically it must have exactly the
$T$th highest valuation.
We will show that player $i$ moving into slot $T$ is a lookahead improving move. Notice that the 
lookahead value for player $i$ staying in slot $i$ is at most the myopic value of staying in that slot.
This follows as the choice of a player two slots below $i$ cannot improve the utility of player $i$ (neither in terms of price
nor slot position), but only could make it worse. 
Hence, it suffices to show that the lookahead value of changing slots is better than the myopic value of staying in slot $i$. 

By several applications of Lemma \ref{sLem1},
we see that player $i$ myopically prefers slot $T$ to slot $i$. 
However, in moving to slot $T$, player $i$ will still make a balanced bid. Thus, no other winning player 
may reduce $i$'s utility by undercutting him. Also, no losing player $j$ wants to move to a winning slot as they can 
only be left with negative utility - since $j$ cannot then be amongst the $T$ highest valuations. 
So moving to slot $T$ is a lookahead improving move for player~$i$.

If player $i$ were originally in slot $T$, then the entire argument can be applied with regards to slots $1$ to $T-1$.
Inductively, we then conclude that in any non-output-truthful equilibrium, there is a lookahead 
improving move, which is a contradiction. This gives us the desired result.
\end{proof}

\subsection{{\sc Average Case Lookahead.}}\ \\
Next, we consider the average-case lookahead model. and show that the above theorem does
not hold for this case. 

\begin{theorem}\label{sthm:avg-badexample}
In GSP auctions, there exist $2$-lookahead equilibria that are not output-truthful in the average-case, leaf model.
\end{theorem}

\begin{proof}
Consider the following example with $n = T = 4$. Let the click-through rates be $c_1 = 35, c_2 = 26, c_3 = 25,$ 
and $c_4 = 20$. Let the valuations be $v_1 = 82, v_2 = 83, v_3 = 100, v_4 = 93$. Starting with the highest slot 
and working to the lowest, let bidder $i$ bid the balanced bid for slot $i$. It can be verified that this turns out to be a 
non-output-truthful equilibria.
\end{proof}

Despite this negative result, $2$-lookahead equilibria cannot have arbitrarily bad social welfare. 

\begin{theorem}\label{sthm:avg-positive}
In GSP auctions, the coordination ratio of $2$-lookahead equilibria is constant in the average-case, leaf model.
\end{theorem}

\begin{proof}
Suppose that we are at an equilibrium. Let $v_{i^*}$ be the $i^{th}$ highest valuation, let player $i^*$ denote the corresponding player, 
let $b_{i^*}$ denote their bid, and $c_{i^*}$ be the click through rate of the slot they currently occupy. We recall that $v_i$ denotes the 
player in slot $i$ and it has click through rate $c_i$ and bid $b_i$.
The social utility of a set $A$ of players is $\sum_{i \in A} {v_ic_i}$.
Thus, by the above definitions, the optimal social utility is $\sum_i {v_{i^*}c_i}$.

Now, choose $\alpha, \beta < 1$ such that $(1-\alpha)^2 > m \beta$. 
Let $I$ be the set of indices $i$ that satisfy both $v_{i} < \alpha v_{i^*}$ and $c_{i^*} < \beta c_{i}$.
Note that for all $i \notin I$ the pair of players $v_i, v_{i^*}$ contribute at least $\min\{\alpha,\beta\}v_{i^*}c_i$ to OPT. So if $I$ is empty, then 
we have achieved a constant coordination ratio. We may thus suppose $I$ is not empty and choose $i \in I$.

Consider $c_{i^*-1}$. As we assume ``balanced" bidding, 
$$b_{i^*} \ge (1-\frac{c_{i^*}}{c_{i^*-1}})v_{i^*}$$
Since $b_{i^*} < b_i < v_i < \alpha v_{i^*}$ by assumption, we have $c_{i^*-1} < \frac{1}{1-\alpha}c_{i^*}$.
Choose $m > 1$. We first prove the following claim.
\begin{claim}\label{scl:I}
For all $i \in I$, we have $c_{i+1} \le \frac{c_i}{m}$.
\end{claim} 
\begin{proof}
Suppose $c_{i+1} > \frac{c_i}{m}$, for some $i\in I$. We will show that player $i^*$ moving into slot $i$ is then lookahead improving. Consider his lookahead 
utility for staying put. Ignoring a repeat move for player $i^*$, which occurs with probability $\frac{1}{n}$, player $i^*$'s utility  in every other 
circumstance is at most $c_{i^*-1}v_{i^*}$, as other players can improve his position by at most one. On the other hand, if player $i^*$ moves 
into slot $i$ then his lookahead utility is at least $c_{i+1}(v_{i^*} - b_{i})$; he wins at least slot $i+1$ and pays at most his bid. 
If player $i$ is chosen to repeat his move then his utility is the same for both cases (as he will then simply 
play a best response move). Thus, it is enough for us to show that
$$c_{i+1}(v_{i^*} - b_{i}) > c_{i^*-1}v_{i^*}$$

However $b_i < v_i < \alpha v_{i^*}$ and putting this together with the above inequalities gives
\begin{eqnarray*}
c_{i+1}(v_{i^*} - b_{i})  &>&  \frac{c_i}{m}(1 - \alpha)v_{i^*} \\
&\ge&  \frac{\beta}{1-\alpha}c_i v_{i^*} \\
&\ge&  \frac{\beta}{1-\alpha}c_{i}v_{i^*} \\
&>&  \frac{1}{1-\alpha}c_{i^*}v_{i^*} \\
&>&  c_{i^*-1}v_{i^*} 
\end{eqnarray*}
We are now done, by our choice of $\alpha$ and $\beta$, and have shown that player $i^*$ moving into slot $i$ is a lookahead improving move. 
This contradicts the fact we are at an equilibria.
\end{proof}

Thus we have established that for all $i \in I$, $c_{i+1} < \frac{c_i}{m}$. Thus, we can bound the optimal social utility contributed by the slots $i \in I$ 
by $\frac{m}{m-1}c_{i_0}v_{i_0*}$ where $i_0 = \min_{i \in I} i$. 

Now if $1 \notin I$ then we have achieved our constant coordination ratio since then either $c_1v_1 > \alpha c_1v_{1^*}$ 
or $c_{1^*}v_{1^*} \ge \beta c_1 v_{1^*}$. Hence, we are guaranteed at least $\min\{\alpha,\beta\}  c_1v_{1^*} \ge \min\{\alpha,\beta\}c_{i_0}v_{i_0*}$,
that is, a least a constant factor of
the social utility from all the slots in $I$ in the optimal allocation. So we suppose $1 \in I$.

Choose $\alpha_1 = \frac{m}{m-1} \alpha$ and consider the player currently in slot $2$. By this choice of $\alpha_1$, 
we ensure that this player does not have value more than $\alpha_1 v_{1^*}$. To see this, recall the player is bidding in a balanced manner 
and so, by Claim \ref{scl:I}, his bid $b_2$ satisfies
$$v_2\ge b_2 \ge (1-\frac{c_2}{c_1})v_2 \ge (1-\frac{1}{m})v_2$$ 
On the other hand, as $1\in I$ we have
$$b_1 = v_1 \le \alpha v_{1^*}$$
Thus, we must have $v_2 \le \frac{m}{m-1} \alpha v_{1^*}= \alpha_1 v_{1^*}$ or the second player would win the first slot.

Now let $\Gamma$ be the set of players with value at least $\alpha_1 v_{1^*}$. Choose some constant $\gamma$. If $|\Gamma| < \gamma n$, then 
player $1^*$'s lookahead utility for moving into slot one is at least $(1-\gamma)(1-\alpha_1)v_{1^*}c_1$. If player $1^*$ stays put, ignoring a repeat 
move for player $1^*$, which occurs with probability $\frac{1}{n}$, player $i^*$'s utility in every other 
circumstance is at most 
$$c_{1^*-1}v_{1^*} < \frac{1}{1-\alpha}c_{1^*}v_{1^*} <  \frac{\beta}{1-\alpha}c_{1}v_{1^*}$$

Since player $1^*$'s utility is the same for both cases when a repeated move occurs and since we can choose $\beta$ sufficiently 
small (i.e, $\beta < (1-\gamma)(1-\alpha)(1-\alpha_1)$), player $1^*$ 
will improve by moving into slot $1$ in this case, contradicting the fact that we are at an equilibrium. 

Thus, we may suppose $|\Gamma| > \gamma n$. Let $i_1 = \max_{i \in \Gamma} i$. Then the players in $\Gamma$ contribute at 
least $\gamma n \alpha_1 v_{1^*}c_{i_1}$ to the social utility. Take a constant $\delta$ 
and suppose that $c_{i_1} \geq \delta \frac{c_1}{n}$. Then the players in 
$\Gamma$ would contribute at least $\gamma \delta \alpha_1 c_1 v_{1^*}$.
Again, this a constant fraction of 
social utility that is contributed in the optimal allocation by player $1^*$ which, in turn, is a constant factor of the optimal social utility of 
the slots in $I$. Thus, we would achieve a constant factor of the optimal social utility. 

So we may assume $c_{i_1} < \delta \frac{c_1}{n}$.  Consider player $i_1$. 
His lookahead utility for staying in place, ignoring the case 
of a repeated move, is at most 
$$
c_{i_1-1}v_{i_1} 
\le \frac{1}{1-\alpha}c_{i_1}v_{i_1}  
\le \frac{1}{1-\alpha}\frac{\delta}{n} c_1 v_{i_1}  
\le \frac{1}{1-\alpha}\frac{\delta}{n} c_1 v_{i^*}  
$$

We may assume that player $v_1\le (1-\epsilon)\alpha_1v_{1^*}$, for some constant $\epsilon$, otherwise we are done. 
Therefore, if player $i_1$ moves 
to slot $1$ then he will earn at least $\epsilon c_1 v_{1^*}$ provided that player $1$
makes the next move. This occurs with probability $1/n$, and so his total lookahead utility, ignoring a repeated move, is at least $\frac{\epsilon}{n}c_1 v_{1^*}$. 
Thus by choosing $\delta \le (1-\alpha)\epsilon$, it follows that the coordination ratio is constant in the 
average case model.  \end{proof}

\section{Industrial Organisation:  Cournot Competition}\label{sec:Cournot}

Next we consider the classical
game theoretic topic of duopolistic competition. Economists have considered a number of alternative models for market 
competition \cite{Tir88}, prominent amongst them is the Cournot model \cite{Cournot}. 
Our main result here is that the social surplus increases when firms are not myopic; surprisingly,
social welfare is actually maximized when firms use $2$-lookahead. 

The Cournot model assumes players sell identical, 
nondifferentiated goods, and studies competition in terms of quantity (rather than price). 
Each player takes turns choosing some quantity of good to produce, $q_i$, and pays some marginal cost to produce it, $c$. The price 
for the good is then set as a function of the quantities produced by both players, $P(q_i + q_j) = (a-q_i-q_j)$, for some constant $a>c$. On turn $l$, each player $i$ makes profit:
$\Pi^l_i(q_i,q_j) = q_i (a-q_i-q_j-c)$.
In this form, the model then has only has one equilibrium, called the {\em Cournot equilibrium}, where $q_i = (a - c)/3$ for each player. 
At equilibrium, each player make a profit of $\Pi^i(q_i,q_j) = q_i(1-2q_i)$. The consumer surplus is $2q_i^2$ and the social surplus is 
then $2q_i(1-q_i)$.

\subsection{{\sc Production under Lookahead Search.}}\ \\
We analyse this game when players apply $k$-lookahead search. 
In industrial settings it is natural to assume that payoffs are collected over time (as in a repeated game); thus, we focus upon the path model.
We define this model inductively. In a $k$-step lookahead path 
model, each player $i$'s utility is the sum of his utilities in the current turn and the $k-1$ subsequent turns. He models the quantities chosen in the 
subsequent turns as though the player acting during those turns were playing the game with a smaller lookahead. More specifically, he assumes 
that the player acting in the $t$'th subsequent turn chooses their quantity to maximise their utility under a $k-t$ lookahead model. 
In order to rewrite this rigorously, let $\pi^i_l$ be the contribution to his utility that player $i$ expects on the $l$th subsequent turn (and $\pi^i_0$ be the 
contribution to his utility that player $i$ expects on his current turn), let $\pi^j_l$ be the contribution to player $j$'s utility that player $i$ expects on the $l$'th 
subsequent turn, and let $q^i_l$ (respectively, $q^j_l$) be the quantity that player $i$ expects to choose (respectively, expects his opponent to choose) under this model.

Then in the path model, player $i$'s expected utility function is $\Pi^i = \sum_{t=0}^{k-1} \pi^i_t$. Player $j$'s expected utility function on player $i$'s turn 
is $\Pi^j = \sum_{t=0}^{k-1} \pi^j_t$. Our aim now is to determine the quantities that player $i$ expects to be chosen by both players in the 
subsequent turns and, thereby, determine the quantity he chooses this turn and the utility he expects to garner. To facilitate the discussion, it should be 
noted that unless noted otherwise, any reference to a ``turn'' refers to a turn during player $i$'s calculation and not an actual game turn.

To simplify our analysis, we will define $q_l$ to be the quantity chosen on turn $l$ by whichever player is acting and $\Pi_l$ to be the expected 
utility that that player garners from turn $l$ to turn $k$. So $\Pi_0 = \Pi^i$, $\Pi_1 = \sum_{t=1}^{k-1} \pi^j_t$, etc. We define $\overline{\Pi}_l$ to be 
the utility garnered from turn $l$ to turn $k$ by the player who does not act during turn $l$. So $\overline{\Pi}_0 = \Pi^j$, $\overline{\Pi}_1 = \sum_{t=1}^{k-1} \pi^i_t$, etc. 
It is clear that on each turn $l$, the active player is trying to maximise $\Pi_l$.

We are now ready to compute these quantities and utilities recursively. We may assume that $a=1$ and $c=0$. By our definition above, we 
have that $\Pi_k = q_k(1-q_k-q_{k-1})$ 
and $\overline{\Pi}_k = q_{k-1}(1-q_k-q_{k-1})$. Our definition also gives us the recursive formula for $l<k$ 
that $\Pi_l = q_l(1-q_l-q_{l-1}) + \overline{\Pi}_{l+1}$ and $\overline{\Pi}_l = q_{l-1}(1-q_l-q_{l-1}) + \Pi_{l+1}$. 
Note that in each of these formulas, $\Pi_l$ and $\overline{\Pi}_l$ are each functions of $q_t$ for $t \geq l$; $q_{l-1}$ is in fact 
fixed on the previous turn and is, therefore, not a variable in $\Pi_l$. 
It is now possible to calculate $q_l$ recursively.
\begin{lemma} The form of $q_l$ is $\beta_l - \alpha_l q_{l-1}$, where $\beta_k = \alpha_k = \beta_{k-1} = \frac{1}{2}, \alpha_{k-1} = \frac{1}{3}$ 
and, for $l<k-1$,
$$
\beta_l =  \frac{2-\beta_{l+1} + \alpha_{l+1}\beta_{l+2} - \alpha_{l+1}\alpha_{l+2}\beta_{l+1}}{4-2\alpha_{l+1} - \alpha_{l+1}^2\alpha_{l+2}}, \ \ 
\alpha_l =  \frac{1}{4-2\alpha_{l+1}-\alpha_{l+1}^2\alpha_{l+2}}
$$
\end{lemma}
\begin{proof}
We proceed by inducting down from $q_k$.
Consider $q_k$ which is the active player's choice on the final turn. As it is the final turn, he is acting myopically and so will 
choose $q_k$ so as to maximise $\Pi_k = q_k(1-q_k-q_{k-1})$. This parobala as a function of $q_k$ is maximised when $q_k = \frac{1-q_{k-1}}{2}$. 
Doing a similar calculation for $\Pi_{k-1} = q_{k-1}(1-q_{k-1}-q_{k-2})+\overline{\Pi}_{k}$ gives us the desired values for $\beta_{k-1}$ and $\alpha_{k-1}$.
We now assume the lemma for all $l > L$ and try to prove it for $q_L$. Recall the recursive formula $\Pi_L = q_L(1-q_L-q_{L-1}) + \overline{\Pi}_{L+1}$. Taking 
the derivative of this with respect to $q_L$  and setting it all equal to zero gives us 
\begin{eqnarray*}
0 &=& (1-2q_L - q_{L-1}) + (1-2q_L - q_{L+1}) 
  - \frac{\partial q_{L+1}}{ \partial q_L} q_L - \frac{\partial q_{L+1}}{\partial q_L} q_{L+2} + \frac{\partial \Pi_{L+2}}{\partial q_{L+2}} \frac{\partial q_{L+2}}{\partial q_L}
\end{eqnarray*}
The last term of the above sum is zero, since $q_{L+2}$ is chosen so that $\frac{\partial \Pi_{L+2}}{q_{L+2}} = 0$. Thus, if 
we plug in the inductive hypothesis into the above equation and simplify, we get
\begin{eqnarray*}
2 &-& \beta_{L+1} + \alpha_{L+1} \beta_{L+2} - \alpha_{L+1} \beta_{L+2} - \alpha_{L+1}\alpha_{L+2}\beta_{L+1} 
= (4-2\alpha_{L+1} - \alpha^2_{L+1}\alpha_{L+2})q_L - q_{L-1}
\end{eqnarray*}
This gives us the desired result.
\end{proof}
Our goal is now to calculate $q_0$ as this will tell us the quantity that player $i$ actually chooses on his turn. From the above lemma, 
we can calculate $q_0$ if we can determine $\alpha_0$ and $\beta_0$. Using numerical methods on the above recursive formula, 
we see that as $k \rightarrow \infty$, $\alpha_0$ decreases towards a limit of $0.2955977\ldots$ and $\beta_0$ approaches a limit 
of $0.4790699\ldots$. These values also converge quite quickly; they both converge to within $0.0001$ of the limiting value for $k \geq 10$. 
Thus, at a lookahead equilibrium, player $i$ will choose $q_i \approx. 0.4790699 - 0.2955977q_j$ and player $j$, symmetrically, 
will choose $q_j \approx 0.4790699 - 0.2955977q_i$. So each player will choose a quantity $q \approx 0.369767$. which is
more than in the myopic equilibrium. Indeed, it is easy to show that for every $k \geq 2$, each player will produce more 
than the myopic equilibrium. This is illustrated in Figure \ref{sfig1}. Observe the quantity produced does not change monotonically 
with the length of foresight $k$, but it does increase significantly if non-myopic lookahead is applied at all. 
Consequently, in the path model looking ahead is better for society overall but worse for each 
individual firm's profitability (as the increase in sales is outweighed by the consequent reduction in price).
\begin{figure}[htb!]
\includegraphics[scale=0.75]{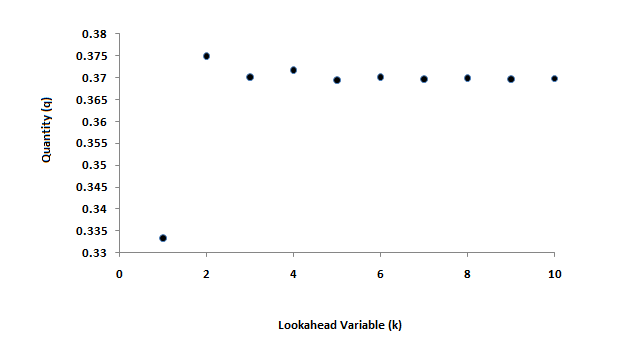}
\caption{How output varies with foresight k}
\label{sfig1}
\end{figure}

\begin{theorem} For Cournot games under the path model, output at a $k$-lookahead equilibrium peaks at $k=2$ with 
output $12.5\%$ larger than at a myopic equilibrium $(k=1)$. As foresight increases, output is $10.9\%$ larger in the limit.
The associated rises in social surplus are $5.5\%$ and $4.9\%$, respectively,
\end{theorem}

\subsection{{\sc Stackelberg Behaviour.}}\ \\
We could also analyse this game under the leaf model, but this model is both 
less realistic here and trivial to analyse. However, it is interesting to note that for the leaf model with asymmetric lookahead, 
where player $i$ has $2$-lookahead and player $j$ has $1$-lookahead, we get the same equilibrium as the classic Stackelberg 
model for competition. Thus, the use of lookahead search can generate leader-follower behaviours.

\section{Unsplittable Selfish Routing}\label{sec:SR}

Now consider the unsplittable selfish routing game. We show that any $2$-lookahead equilibrium has a constant
coordination ratio. We then show how to derive a similar result for $2$-lookahead dynamics. 

For this game we have a directed graph $G = (V,E)$ and a set of $n$ agents. Agent $i$ wants to route $1$ unit of flow from a source $s_i$ to a destination $t_i$. 
Each agent $i$ chooses an $s_i-t_i$ path $P_i$ and these paths together generate a flow $f$. We assume that there 
is a linear {\em latency function} $\lambda_e(f_e) = a_ef_e + b_e$ on each edge edge $e\in E$. 
The total latency of a flow $f$ is denoted $l(f) = \sum_{e \in E} \lambda_e(f_e) f_e = (a_ef_e + b_e)f_e$. 
The latency of player $i$ is denoted $l_i(f) = \sum_{e \in P_i} a_ef_e + b_e$; observe that $l(f) = \sum_{i \in U} l_i(f)$.  
For this game, we consider $2$-lookahead in both the leaf and path models, under the average-case lookahead model.

Recall, in the leaf model, a player $i$'s move from a flow $f$ to a flow $f'$ is \textit{lookahead improving} if $E(l_i(f'')|f') > E(l_i(f'')|f)$ 
where $f''$ is the flow obtained after the next player (chosen uniformly at random amongst all the players) makes a (myopic) best response. 
In the path model a player $i$'s move from a flow $f$ to a flow $f'$ is \textit{lookahead improving} if 
$\frac{1}{2}l_i(f') + \frac{1}{2}E(l_i(f'')|f') > \frac{1}{2}l_i(f) + \frac{1}{2}E(f''|f)$ where $f''$ is as above.



\begin{theorem}\label{sthm:cong-eq-constant}
In the average-case $2$-lookahead leaf model, the coordination ratio for an equilibrium is at most $(1 + \sqrt{5})^2$.
\end{theorem}

\begin{proof}
This proof adapts the result in \cite{AAE05} to our setting.
Let $f$ be any flow at a lookahead equilibrium and $f^*$ be an optimal flow. Suppose player $i$ is 
taking path $P_j$ in flow $f$ and path $P_j^*$ in flow $f^*$. 
Let $J(e)$ be the set of players using edge $e$ in the flow $f$ and let $J^*(e)$ be the same for $f^*$.

At a lookahead equilibrium, player $j$ doesn't want to move from $P_j$ to $P_j^*$. 
This means that after a random/worst case next move,  the strategy $P_j$ has a higher (expected) payoff than
the strategy $P_j^*$. In particular, it must the case that the best possible outcome resulting from from choosing 
$P_j$ has a higher (expected) payoff than the worst possible outcome resulting from the strategy $P_j^*$.
In the former case, the best possible outcome is that the next player had also been using the path 
$P_j$ but then moves completely off the path.
Similarly, in the latter case, the worst possible outcome is that the next player had not been using any edge on the path 
$P_j^*$ but then changes strategy and also selects the path $P_j^*$ entirely.
Thus we must have:

\begin{eqnarray*}
  \sum_{e \in P_j^*} a_e(f_e + 2) + b_e   &\ge& \sum_{e \in P_j} a_ef_e + b_e - \sum_{e \in P_j: f_e \geq 2} a_e
  \end{eqnarray*}

Summing over all players $j$, we obtain
\begin{eqnarray*}
\sum_j   \sum_{e \in P_j^*} a_e(f_e + 2) + b_e   &\ge&  \sum_j \left(\sum_{e \in P_j} a_ef_e + b_e - \sum_{e \in P_j: f_e \geq 2} a_e \right) \\
  &=&\sum_{e \in E} \sum_{j \in J(e)} a_ef_e + b_e - \sum_j \sum_{e \in P_j: f_e \geq 2} a_e \\
&=& \sum_{e \in E} (a_ef_e + b_e) f_e - \sum_{e \in P_j: f_e \geq 2} a_ef_e \\
&\ge& \sum_{e \in E} (a_ef_e + b_e) f_e- \sum_{e \in P_j: f_e \geq 2} \frac12 a_ef_e^2 \\
&\ge& \sum_{e \in E} \frac12 (a_ef_e + b_e) f_e \\
&=& \frac12 \sum_{e \in E} \lambda_e(f_e) 
\end{eqnarray*}

Rearranging gives and applying the Cauchy-Schwartz inequality\footnote{{\small For any two vectors ${\bf x}$ and ${\bf y}$, 
we have ${\bf x}^T{\bf y} \le \sqrt{{\bf x}^T{\bf x}} \cdot \sqrt{{\bf y}^T{\bf y}}$.}} produces
\begin{eqnarray*}
\frac12 \sum_{e \in E} \lambda_e(f_e) &\le&    \sum_j   \sum_{e \in P_j^*} a_e(f_e + 2) + b_e     \\
 & =& \sum_{e \in E} (a_e(f_e + 2) + b_e)f^*_e  \\
& \le& \sum_{e \in E} a_ef_ef_e^* + (2a_e+b_e)f_e^*  \\
& \le & \sum_{e \in E} a_ef_ef_e^* + 2\lambda_e(f_e^*)  \\
& \le &  \sqrt{\sum_{e \in E} a_ef_e^2}\cdot \sqrt{\sum_{e \in E} a_e{f_e^*}^2} + 2\sum_{e \in E} \lambda_e(f_e^*) \\
& \le &  \sqrt{\sum_{e \in E} \lambda_e(f_e)}\cdot \sqrt{\sum_{e \in E} \lambda_e(f_e^*)} + 2\sum_{e \in E} \lambda_e(f_e^*)  
\end{eqnarray*}

Set 
$\rho = \sqrt{\frac{\sum_e \lambda_e(f_e)}{\sum_e \lambda_e(f_e^*)}}$
and observe that $\rho^2$ is the coordination ratio, given we choose the worst lookahead equilibrium $f$.
Consequently,  $\frac{1}{2}\rho^2 \leq \rho + 2$.
Solving gives $\rho \le 1+\sqrt{5}$ as desired.
\end{proof}

Next we consider the lookahead dynamics and study coordination ratio
for the lookahead dynamics.

\begin{theorem} \label{sthm1}
In the average-case $2$-lookahead model, the coordination ratio for lookahead dynamics is a constant for the leaf model.
\end{theorem}
\begin{proof} 

We follow a similar approach to Theorem 4.1 in \cite{GMV05} and start by proving some sub-lemmas.

\begin{lemma}\label{sFundIneq}
If player $i$ makes a lookahead improving move from path $P_i$ to $P_i'$ which changes the flow 
from $f$ to $f_i'$ then $l_i(f_i') \leq 2l_i(f) + \frac{1}{n}l(f)$.
\end{lemma}

\begin{proof}
So player $i$'s lookahead cost with $f_i'$ is less than his cost with $f$. Moreover, we can lower bound the 
lookahead cost of $f_i'$ by the quantity 
\begin{eqnarray*}
\lefteqn{ \sum_{e \in P_i'} \frac{f_e}{n} (a_ef_e + b_e) + (1-\frac{f_e}{n})(a_e(f_e+1) + b_e) }\\
 &=& \sum_{e \in P_i'} a_e(f_e+1) + b_e - \frac{f_e}{n} a_e \\
  &\ge& \sum_{e \in P_i'} (1-\frac{1}{n})a_e(f_e+1) + b_e  \\
 &\ge& \sum_{e \in P_i'} (1-\frac{1}{n})(a_e(f_e+1) + b_e)  \\
&=& (1-\frac{1}{n})l_i(f_i')
\end{eqnarray*}
This would be the cost incurred if the randomly selected next player $j$ avoids
any edge $e$ that player $i$ is on (either by moving away from $e$ or not moving onto $e$).
Using similar reasoning, we may upper bound the cost to player $i$ of sticking with $P_i$ by 
\begin{eqnarray*}
\lefteqn{\sum_{e \in P_i} \frac{f_e}{n} (a_ef_e + b_e) + (1-\frac{f_e}{n})(a_e(f_e+1) + b_e) }\\
&=&  \sum_{e \in P_i}  (a_ef_e + b_e) + (1-\frac{f_e}{n}) a_e \\ 
&\le &  \sum_{e \in P_i}  a_e(f_e+1) + b_e  
\ \le\    \sum_{e \in P_i}  2a_ef_e + b_e  \\
&\le& 2 l_i(f) 
\end{eqnarray*}
Here we assumed the next player $j$ selects every edge $e$ that player $i$ is on
(either by staying on $e$ or by moving onto $e$). Therefore,
$l_i(f_i') \le 2(1+\frac{1}{n-1})l_i(f)$ 
which implies the statement in the lemma.
\end{proof}

  Applying  Lemma \ref{sFundIneq} with Lemma 4.2 in \cite{GMV05}, 
we get:

\begin{lemma} \label{slem1}
If agent $i$ changes his path from $P_i$ to $P_i'$, changing the flow from $f$ to $f_i'$, then $l(f_i') \leq l(f) + (d+1)l_i(f_i')-l_i(f)$. 
In particular, if agent $i$ makes a lookahead improving move then $l(f_i') \leq (1+\frac{1}{n})l(f) + 3l_i(f)$.
\end{lemma}

Now, applying  Lemma \ref{slem1} with Lemma 4.3 in \cite{GMV05}.

\begin{lemma}
Let $f$ be the current flow. Suppose we chose a player at random and they make a lookahead best response 
resulting in flow $f'$. Then $E(l_i(f')|f) \leq (1+\frac{4}{n})l(f)$.
\end{lemma}
Finally, we prove the following lemma which will imply the statement of the theorem. 

\begin{lemma}\label{slastlemma} Let $f$ be the current flow. Suppose we chose a player at random and they make a lookahead best 
response resulting in flow $f'$. Then either $E(l(f') | f) \leq (1 - \frac{1}{2n})l(f)$ or $l(f) < (6 + \sqrt{37})OPT$.
\end{lemma}

\begin{proof}

Suppose player $i$ changes his path from $P_j$ to $P_j'$ resulting in the flow changing from $f$ to $f_i'$. 
Thus $E(l(f')| f)  =  \frac{1}{n} \sum_{i} l(f_i')$.

Case 1: $\sum_{i} 4 l_i(f_i') \leq \sum_{i} l_i(f)$
\begin{eqnarray*}
E(l(f') | f) & = & \frac{1}{n} \sum_{i} l(f_i') \\
& \leq & \frac{1}{n} \sum_{i} l(f) + l_i(f_i') - l_i(f) + \sum_{e \in P_i'-P_i} a_ef_{i,e} \\
& \leq & \frac{1}{n} \sum_{i} l(f) + 2l_i(f_i') - l_i(f) \\
& \leq & \frac{1}{n} \sum_{i} l(f) + \frac{1}{2}l_i(f) - l_i(f) \\
& = & (1 - \frac{1}{2n})l(f)
\end{eqnarray*}

Case 2: $\sum_{i} 4 l_i(f_i') > \sum_{i} l_i(f) = l(f)$

Let $f^*$ be the optimal flow and let $P_i^*$ be player $i$'s path in this flow. Let $J^*(e)$ be the set of players 
on edge $e$ in $f^*$. Since $P_i'$ is a lookahead best response, we may apply Lemma \ref{sFundIneq} to see 
that $l_i(f_i') \leq 2l_i(f^*) + \frac{1}{n}l(f^*)$. Thus
\begin{eqnarray*}
l(f) & <  & 4 \sum_{i} l_i(f_i') 
\ \leq\   4 \sum_{i} 2 l_i(f^*) + \frac{1}{n}l(f^*) \\
& = & 12 \sum_{i} l_i(f^*) 
\ =\   12 \sum_{i} \sum_{e \in E} a_e f_e^* + b_e \\
& \leq & 12 \sum_{e \in E} \sum_{i \in J*(e)} a_e (f_e + 1) + b_e \\
& = & 12 \sum_{e \in E} a_e f_ef_e^* b_ef_e^* + a_ef_e^* \\
& \leq & 12 \sqrt{l(f)l(f^*)} + l(f^*)
\end{eqnarray*}

where the last inequality follows from Cauchy-Schwartz. Thus, if we set $x = \sqrt{\frac{l(f)}{OPT}}$, the above 
can be transformed into the inequality
$x^2 \leq 12 x + 1$.
\end{proof}

The remainder of the proof of Theorem \ref{sthm1} follows by applying the above lemmas as shown in \cite{GMV05}.
 \end{proof}


\subsection{{\sc Valid Utility Games.}}\ \label{sec:UG} \\
Here is a bad example for the path model (a slightly modified example applies to the leaf model). 
It applies for any number $t$ of lookahead moves.
Take a Steiner Set System $S(2, k, n)$. For example, these exist with $n= q^2+q+1$ and $k=q+1$.
Let each subset in the system induce a "sub-game" - thus each pair of players are
together in exactly one subgame. Consequently, each player is in $\frac{n-1}{k-1}= q+1=k$ subgames, and $n$ games
in total.
The strategy set of a player $i$ in subgame $g$ is $\{y_i^g, x_{i,1}^g, x_{i,2}^g, \dots,  x_{i,k}^g\}$.
It has one {\em nice} strategy and $k$ {\em naughty} strategies:
player $i$ always gets one point for playing the nice strategy $y_i^g$, but gets two points for playing a naughty strategy 
$x_{i,l_i}^g$ {\em provided}
$\sum_{j} l_j = i\mod k$, where the sum is over all players $j$ who are playing a strategy $x^g_{j,l_j}$ - we call
$i$ the winner of subgame $g$ in the case. 

Thus a player $i$ who moves next can guarantee $k$ points by playing $y$s but can guarantee
$2k$ points by playing $x^g$s to win all $k$ subgames it is in. Moreover,
the player can lose at most one game in each subsequent time period. This follows
as the next $t=k$ players share exactly one game each with player $i$.
Thus the player, in the worst case receives $2k+2(k-1)+\cdots+4+2= k(k+1)$ in the next $k$ moves.
This is greater than the $k^2$ payoff from playing only $y$s.

Consider then the dynamics of this game under $k$-lookahead search.
Over time, at any state of play, the total value of the game will be $2n$; in each of the $n$ subgames
all the players are behaving naughtily. 
The optimal value however is $n(k+1)$;  in each subgame, $k-1$ of the players are nice and one is naughty.
So we have shown:
\begin{lemma}
For valid utility games, in the path model the coordination ratio of $k$-lookahead dynamics 
is at least $\frac{k+1}{2} = \frac{t+1}{2} \ge \frac12 \sqrt{n}$. \qed
\end{lemma}
 
\subsection{{\sc Basic Utility Games.}}\ \\
For basic utility games, good guarantees can be obtained for the path model. More interestingly, 
for the leaf model lookahead equilibria can be extremely bad, even for $2$-lookahead equilibria.

\begin{lemma}\label{sbasicutility}
In  basic utility games, the coordination ratio of $2$-lookahead equilibria
can be arbitrarily bad in the leaf model.
\end{lemma}

\begin{proof}
Consider the following symmetric 2-player game. Let each player have a groundset $\{B,T,G\}$.
A feasible strategy consists of playing at most one action in the groundset.
We create a submodular social function using the table
\begin{center}
\begin{tabular}{c|cccc}
&  $\emptyset$ & B & T &G \\
\hline 
B & 6 & 6 & 6 & 1\\
T &  $\kappa$-9 & $\kappa$-9 & 7&4\\
G &  $\kappa$-5 &  $\kappa$-10& 8&5
\end{tabular}
\end{center}
Set $\gamma(\emptyset, \emptyset) =0$. Then
let the $ij$th entry of the matrix, $\delta_{ij}$, be the marginal value of adding action $i$ when action $j$ is
being played by the other player. 
For example, $\gamma(B,\emptyset) = \gamma(\emptyset, \emptyset) +\delta_{B,\emptyset}= 0+6=6$.
Similarly, $\gamma(B,B)=12,\gamma(T,\emptyset)= \kappa-9, \gamma(G,\emptyset)=\kappa-5,
\gamma(B,G)=\kappa-4,\gamma(B,T)=\kappa-3,\gamma(T,T)=\kappa-2,\gamma(T,G)=\kappa-1,\gamma(G,G)=\kappa$.

We need to extend this definition to all subsets.
Suppose that Player $1$ is currently choosing $S_1$ and
Player $2$ is currently choosing $S_2$.
To complete the definition of $\gamma$, we say that the marginal
value of adding action $i$ to the subset $S=S_1\cup S_2$, is
$\delta_{i,S} = \min_{j\in S_1\cup S_2} \delta_{ij}$.

Note that this is true if $i$ is added to $S_1$ and if $i$ is added to $S_2$.
This processes produces a submodular social function.
The payoff functions are then defined in accordance with the Vickrey condition.

Clearly, as the players are constrained to play singleton actions, the optimal
solution $\Omega=\{G,G\}$ has value $\kappa$.
We claim that $\{B,B\}$, with social value $12$, is the {\em only} equilibrium in the leaf model.
Thus, for any $\kappa$, we can be a factor $\Omega(\kappa)$ away from the optimal social value.

To prove this, first suppose that Player $1$ plays $B$. According to the Vickrey condition, the best response
of Player $2$ is to play $T$ (she needs to choose $*$ maximize $\gamma(B,*)$). The payoff to player $1$ is then
$\gamma(B,T) - \gamma(\emptyset, T) = (\kappa-3)-(\kappa-9) = 6$.
Second suppose that Player $1$ plays $T$. According to the Vickrey condition, the best response
of Player $2$ is to play $G$ (she needs to maximize $\gamma(T,*)$). The payoff to player $1$ is then
$\gamma(T,G) - \gamma(\emptyset, G) = (\kappa-1)-(\kappa-5) = 4$.
Finally suppose that Player $1$ plays $G$. According to the Vickrey condition, the best response
of Player $2$ is to play $G$ - observe this must be the case as $(G,G)$ is the optimal solution. 
The payoff to player $1$ is then
$\gamma(G,G) - \gamma(\emptyset, G) = \kappa-(\kappa-5) = 5$.

Thus, with $2$-lookahead, Player $1$ will always think it in his interest to play $B$.
(Note that in the leaf model, it is irrelevant for Player $1$ what strategy Player $2$ is currently 
playing.) By a symmetric argument, Player $2$ will always think it in her interest to play $B$.
\end{proof}

\section{Shapley Network Design Games}\label{sec:Shapley}
For our final example we show that the use of lookahead search may 
allow for ``uncoordinated'' cooperative behaviours. By looking ahead, a player may select a cooperative move
whose consequence can be to induce other players to also make cooperative moves. 
We give a very simple illustration of this behaviour. Consider the following 
Shapley network design game:
Given a network, there is a single source $s$ and a single sink $t$.
We have $n$ players, each wanting to route from $s$ to $t$. There are $N$ paths (where $N$ may be exponential) to choose from.
The cost of any link is equally shared between those players that use it.
The coordination ratio is then easily seen to be at least $n$.
However, the coordination ratio improves by a factor $k$, when the players use $k$-lookahead search.

\begin{theorem}\label{sthm:shap}
The coordination ratio of $k$-lookahead dynamics for Shapley network design games in the leaf model is at most $n/k$.
\end{theorem}

\begin{proof}
We present the proof for the worst-case lookahead model. The proof for the
average-case model uses the same idea.
Assume the players are currently choosing the paths $\{\bar{P}_1,\bar{P}_2,\dots,\bar{P}_n\}$.
Consider the depth $k$ tree when the players move in the order $1,2,\dots, k$.
Take a decision node for player $k-1$. This has $N$ children that are decision nodes for 
player $k$. Let the paths chosen by player $k$ at these nodes be $Q_1,Q_2,\dots,Q_N$, respectively.
Suppose that in response to this move, player $k-1$ chooses the path $P_{j}$. We claim that
$P_j=Q_j$.
\begin{eqnarray*}
c(P,\mathcal{T'}) 
&=& 
\sum_{e\in P/(Q_j\cup P_j)} \frac{c_e}{n_e}  +   \sum_{e\in P\cap Q_j\cap P_j} \frac{c_e}{n_e} \\
&& +
 \sum_{e\in (P\cap Q_j)/P_j} \frac{c_e}{n_e+1} +  \sum_{e\in (P\cap P_j)/Q_j} \frac{c_e}{n_e-1} \\
&=& c(P,\mathcal{T}) -  \sum_{e\in (P\cap Q_j)/P_j} \left( \frac{c_e}{n_e} -  \frac{c_e}{n_e+1}\right) \\
&&
+  \sum_{e\in (P\cap P_j)/Q_j} \left( \frac{c_e}{n_e-1} -  \frac{c_e}{n_e}\right)\\ 
\end{eqnarray*}
Thus
$$ 
c(Q_j,\mathcal{T'}) = c(Q_j,\mathcal{T})  - \sum_{e\in Q_j/P_j} \left( \frac{c_e}{n_e} -  \frac{c_e}{n_e+1}\right)
$$
Now since $c(Q_j,\mathcal{T}) \le  c(P,\mathcal{T})$ we have
\begin{eqnarray*}
c(P,\mathcal{T'}) 
&=& 
c(P,\mathcal{T}) -  \sum_{e\in (P\cap Q_j)/P_j} \left( \frac{c_e}{n_e} -  \frac{c_e}{n_e+1}\right)\\
&&
+  \sum_{e\in (P\cap P_j)/Q_j} \left( \frac{c_e}{n_e-1} -  \frac{c_e}{n_e}\right)\\ 
&\ge& c(Q_j,\mathcal{T}) -  \sum_{e\in (P\cap Q_j)/P_j} \left( \frac{c_e}{n_e} -  \frac{c_e}{n_e+1}\right)\\
&&
+  \sum_{e\in (P\cap P_j)/Q_j} \left( \frac{c_e}{n_e-1} -  \frac{c_e}{n_e}\right)\\ 
&\ge& c(Q_j,\mathcal{T}) -  \sum_{e\in (P\cap Q_j)/P_j} \left( \frac{c_e}{n_e} -  \frac{c_e}{n_e+1}\right)  \\ 
&\ge& c(Q_j,\mathcal{T}) -  \sum_{e\in Q_j/P_j} \left( \frac{c_e}{n_e} -  \frac{c_e}{n_e+1}\right)  \\ 
&=& c(Q_j,\mathcal{T'}) 
\end{eqnarray*}

This proves the claim.
Applying induction, we see that each player $1,\dots,k$ will play the same strategy $P^*$, and thus,
receive the same payoff.
Let's take the worst case choice for players $2,\dots,k$ from the point of view of player $1$.
If $P^*=P^{SP}$, the shortest $s-t$ path, then each of the $k$ chosen
players will have a cost of at most

$$\frac{c(P^{SP})}{k} \le \frac{\opt}{k}$$

Thus, if $P^*\neq P^{SP}$, then player $1$ can guarantee himself a cost of at most
$\frac{\opt}{k}$. This argument applies for all players so, in an equilibrium,
the total cost is at most, $\frac{n}{k} \opt$. 
\end{proof}

\end{document}